\documentclass[10pt,conference]{IEEEtran}
\usepackage{url}

\usepackage{amsmath}
\usepackage{amssymb}

\usepackage{theorem}

\theoremstyle{plain}
\newtheorem{theorem}{Theorem}

\newtheorem{corollary}[theorem]{Corollary}

\newtheorem{definition}[theorem]{Definition}

\newcommand{\abs}[1]{\left\lvert#1\right\rvert}
\DeclareMathOperator{\Sim}{sim}
\DeclareMathOperator{\Dom}{dom}
\DeclareMathOperator*{\Limsup}{\varlimsup}

\newcommand{\N}{\mathbb{N}}%      \N   == \mathbb{N}
\newcommand{\Q}{\mathbb{Q}}%      \Q   == \mathbb{Q}
\newcommand{\R}{\mathbb{R}}%      \R   == \mathbb{R}
\newcommand{\X}{\{0,1\}^*}%        \X  == \Sigma^*
\newcommand{\XI}{\{0,1\}^\infty}%        \XI  == \Sigma^\infty
\newcommand{\K}{K}

\begin{document}

\title{The Tsallis entropy and the Shannon entropy of\\
a universal probability}

%\author{\authorblockN{Helmut B\"{o}lcskei}
%\authorblockA{Communication Technology Laboratory\\
%ETH Z\"{u}rich\\
%CH-8092 Z\"{u}rich, Switzerland\\
%Email: boelcskei@nari.ee.ethz.ch}}
%\and
%\authorblockN{Ralf Koetter}
%\authorblockA{Department of Electrical Engineering\\
%University of Illinois at Urbana-Champaign\\
%Urbana, IL 61801, USA\\
%Email: koetter@uiuc.edu}
%\and
%\authorblockN{Gerhard Kramer}
%\authorblockA{Bell Laboratories, Alcatel-Lucent\\
%Murray Hill, NJ 07974-0636, USA\\
%Email: gkr@research.bell-labs.com}}
% avoiding spaces at the end of the author lines is not a problem with
% conference papers because we don't use \thanks or \IEEEmembership
%Homer Simpson\authorrefmark{2},
%Montgomery Scott\authorrefmark{3} and
%Eldon Tyrell\authorrefmark{4}}
%Telephone: (800) 555--1212, Fax: (888) 555--1212}
%\authorblockA{\authorrefmark{4}Tyrell Inc., 123 Replicant Street, Los Angeles, California 90210--4321}}
\author{\authorblockN{Kohtaro Tadaki}
\authorblockA{Research and Development Initiative, Chuo University\\
1-13-27 Kasuga, Bunkyo-ku, Tokyo 112-8551, Japan\\
Email: tadaki@kc.chuo-u.ac.jp}}

\maketitle

\begin{abstract}
We study the properties of Tsallis entropy and Shannon entropy
from the point of view of algorithmic randomness.
In algorithmic information theory,
there are two equivalent ways to define the program-size complexity $\K(s)$ of
a given finite binary string $s$.
In the standard way, $\K(s)$ is defined as the length of
the shortest input string
for the universal self-delimiting Turing machine to output $s$.
In the other way,
the so-called universal probability $m$ is introduced first,
and then $\K(s)$ is defined as $-\log_2 m(s)$
without reference to the concept of program-size.
In this paper,
we investigate
the properties of
the Shannon entropy, the power sum, and the Tsallis entropy of
a universal probability
by means of the notion of program-size complexity.
We determine the convergence or divergence of each of
these three quantities,
and
evaluate
its degree of randomness
if it converges.
\end{abstract}

%%%%%%%%%%%%%%%%%%%%%%%%%%%%%%%%%%%%%%%%%%%%%%%%%%%%%%%%%%%%%%%%%%%%%%%%%%%
%% start the paper here:
\section{Introduction}

Algorithmic information theory is a framework to apply
information-theoretic and probabilistic ideas to recursive function theory.
One of the primary concepts of algorithmic information theory
is the \textit{program-size complexity}
(or \textit{Kolmogorov complexity}) $\K(s)$ of a finite binary string $s$,
which is defined as the length of the shortest binary
program
for the universal self-delimiting Turing machine $U$ to output $s$.
By the definition,
$\K(s)$ can be thought of as the information content of
the individual finite binary string $s$.
In fact,
algorithmic information theory has precisely the formal properties of
classical information theory (see Chaitin \cite{C75}).
The concept of program-size complexity plays a crucial role in
characterizing the randomness of a finite or infinite binary string.

The program-size complexity $\K(s)$ is originally defined
using the concept of program-size, as stated above.
However,
it is possible to define $\K(s)$ without referring to such a concept,
i.e.,
we first introduce a \textit{universal probability} $m$,
and then define $\K(s)$ as $-\log_2 m(s)$.

In this paper,
we investigate
the properties of
the Shannon entropy, the power sum, and the Tsallis entropy of
a universal probability,
from the point of view of algorithmic randomness,
by means of the notion of program-size complexity.
In particular,
we show
the following:
(i) The Shannon entropy of any universal probability diverges to infinity.
(ii) If $q$ is a
computable
real number with $q\ge 1$,
then the power sum $\sum_s m(s)^q$ of any universal probability $m$
has the \textit{degree of randomness} at least $1/q$.
Here the notion of degree of randomness is
a stronger notion than compression rate,
and is defined using the program-size complexity \cite{T99,T02}.
(iii) If $0<q<1$,
then the power sum $\sum_s m(s)^q$
diverges to infinity.
(iv) In the case where $q$ is a computable real number with $q>1$,
the Tsallis entropy $S_q(m)$ of a universal probability $m$
can have any
computable
degree of randomness.
(v) If $0<q<1$,
then the Tsallis entropy $S_q(m)$ diverges to infinity.

%%%%%%%%%%%%%%%%%%%%%%%%%%%%%%%%%%%%%%%%%%%%%%%%%%%%%%%%%%%%%%%%%%%%%%%%%%%
\section{Preliminaries}
\label{preliminaries}

We start with some notation about numbers and strings
which will be used in this paper.

$\N=\left\{0,1,2,3,\dotsc\right\}$
is the set of natural numbers,
and $\N^+$ is the set of positive integers.
$\Q$ is the set of rational numbers,
and $\R$ is the set of real numbers.
$\X =
\left\{
  \lambda,0,1,00,01,10,11,000,001,010,\dotsc
\right\}$
is the set of finite binary strings
where $\lambda$ denotes the \textit{empty string},
and $\X$ is ordered as indicated.
We identify any string in $\X$ with a natural number in this order,
i.e.,
we consider $\varphi\colon \X\to\N$ such that $\varphi(s)=1s-1$
where the concatenation $1s$ of strings $1$ and $s$ is regarded
as a dyadic integer,
and then we identify $s$ with $\varphi(s)$.
For any $s \in \X$, $\abs{s}$ is the \textit{length} of $s$.
A subset $S$ of $\X$ is called a \textit{prefix-free set}
if no string in $S$ is a prefix of another string in $S$.
$\XI$ is the set of infinite binary strings,
where an infinite binary string is
infinite to the right but finite to the left.
For any $\alpha\in \XI$ and any $n\in\N^+$,
$\alpha_n$ is the prefix of $\alpha$ of length $n$.
For any partial function $f$,
the domain of definition of $f$ is denoted by $\Dom f$.
We write ``r.e.'' instead of ``recursively enumerable.''

Normally, $o(n)$ denotes any function $f\colon \N^+\to\R$ such
that $\lim_{n \to \infty}f(n)/n=0$.
On the other hand,
$O(1)$ denotes any function $g\colon \N^+\to\R$ such that
there is $C\in\R$ with the property that
$\abs{g(n)}\le C$ for all $n\in\N^+$.

Let $T$ be an arbitrary real number.
$T\bmod 1$ denotes $T - \lfloor T \rfloor$,
where $\lfloor T \rfloor$ is the greatest integer less than or equal to $T$.
Hence, $T\bmod 1 \in [0,1)$.
We identify a real number $T$ with
the infinite binary string $\alpha$ such that
$0.\alpha$ is the base-two expansion of $T\bmod 1$ with infinitely many zeros.
Thus, $T_n$ denotes the first $n$ bits of
the base-two expansion of
the real number
$T\bmod 1$ with infinitely many zeros.

We say that a real number $T$ is \textit{computable} if
there exists a total recursive function $f\colon\N^+ \to \Q$ such that
$\abs{T-f(n)} < 2^{-n}$
for all $n\in\N^+$.
We say that $T$ is \textit{right-computable} if
there exists a total recursive function $g\colon\N^+\to\Q$ such that
$T\le g(n)$ for all $n\in\N^+$ and $\lim_{n\to\infty} g(n)=T$.
We say that $T$ is \textit{left-computable} if $-T$ is right-computable.
It is then easy to see that,
for any $T\in\R$,
$T$ is computable if and only if
$T$ is both right-computable and left-computable.
See e.g.~Pour-El and Richards \cite{PR89} and Weihrauch \cite{W00} 
for the detail of the treatment of
the computability of real numbers and real functions on a discrete set.

%%%%%%%%%%%%%%%%%%%%%%%%%%%%%%%%%%%%%%%%%%%%%%%%%%%%%%%
\subsection{Algorithmic information theory}
\label{ait}

In the following
we concisely review some definitions and results of
algorithmic information theory
\cite{C75,C87a,C87b}.
A \textit{computer} is a partial recursive function
$C\colon \X\to \X$
such that
$\Dom C$ is a prefix-free set.
For each computer $C$ and each $s \in \X$,
$\K_C(s)$ is defined by
$\K_C(s) =
\min
\left\{\,
  \abs{p}\,\big|\;p \in \X\>\&\>C(p)=s
\,\right\}$.
A computer $U$ is said to be \textit{optimal} if
for each computer $C$ there exists a constant $\Sim(C)$
with the following property;
if $C(p)$ is defined, then there is a $p'$ for which
$U(p')=C(p)$ and $\abs{p'}\le\abs{p}+\Sim(C)$.
It is easy to see that there exists an optimal computer.
Note that the class of optimal computers equals to
the class of functions which are computed
by \textit{universal self-delimiting Turing machines}
(see \cite{C75} for the detail).
We choose
a particular
optimal computer $U$ as the standard one for use,
and define $\K(s)$ as $\K_U(s)$,
which is referred to as
the \textit{program-size complexity} of $s$,
the \textit{information content} of $s$, or
the \textit{Kolmogorov complexity} of $s$.
Thus, $\K(s) \le \K_C(s) + \Sim(C)$ for any computer $C$.

The program-size complexity $\K(s)$ is originally defined
using the concept of program-size, as stated above.
However,
it is possible to define $\K(s)$ without referring to such a concept,
i.e.,
as in the following,
we first introduce a \textit{universal probability} $m$,
and then define $\K(s)$ as $-\log_2 m(s)$.
We say that $r$ is a \textit{semi-measure on $\X$}
if $r\colon \X\to[0,1]$ such that $\sum_{s\in \X}r(s)\le 1$.
A universal probability is defined as follows \cite{ZL70}.

%%%%% for compression
\vspace*{-1.0mm}

\begin{definition}[universal probability]
  We say that $r$ is a \textit{lower-computable semi-measure} if
  $r$ is a semi-measure on $\X$ and
  there exists a total recursive function
  $f\colon\N^+\times \X\to\Q$
  such that, for each $s\in \X$,
  $\lim_{n\to\infty} f(n,s)=r(s)$ and
  $\forall\,n\in\N^+\;\>0\le f(n,s)\le r(s)$.
  We say that a lower-computable semi-measure $m$ is
  a \textit{universal probability} if
  for any lower-computable semi-measure $r$,
  there exists a real number $c>0$ such that,
  for all $s\in \X$, $c\,r(s)\le m(s)$.
\hfill\QED
\end{definition}

The following theorem can be then shown
(see e.g.~Theorem 3.4 of Chaitin \cite{C75} for its proof).
Here, $P(s)$ is defined as $\sum_{U(p)=s}2^{-\abs{p}}$
for each $s\in \X$.

%%%%% for compression
\vspace*{-2.0mm}

\begin{theorem}\label{eup}
  Both $2^{-\K(s)}$ and $P(s)$ are universal \\
  probabilities.
  \hfill\QED
\end{theorem}

By Theorem \ref{eup}, we see that, for any universal probability $m$,
\begin{equation}\label{eq: K_m}
  \K(s)=-\log_2 m(s)+O(1).
\end{equation}

Thus it is possible to define $\K(s)$ as $-\log_2 m(s)$
with
a particular
universal probability $m$ instead of as $\K_U(s)$.
Note that
the difference up to an additive constant is
nonessential
to algorithmic information theory.
Any universal probability is not computable,
as corresponds to the uncomputability of $\K(s)$.
As a result, we see that
$0<\sum_{s\in\X}m(s)<1$ for any universal probability $m$.

For any $\alpha\in\XI$,
we say that $\alpha$ is
\textit{weakly Chaitin random}
if there exists $c\in\N$ such that,
for all $n\in\N^+$,
$n-c\le \K(\alpha_n)$ 
\cite{C75,C87b}.
As the total sum of the universal probability $2^{-\K(s)}$,
Chaitin \cite{C87a} introduced the real number $\theta$ by
%%%%% for compression
%\vspace*{-1.0mm}
\begin{equation}\label{theta}
  \theta=\sum_{s\in\X}2^{-\K(s)}.
\end{equation}
%%%%% for compression
%\vspace*{-2mm}\\
Then \cite{C87a} showed that $\theta$ is weakly Chaitin random.

In the works \cite{T99,T02},
we generalized the notion of
the randomness of an infinite binary string
so that the degree of the randomness can be characterized
by a real number $D$ with $0<D\le 1$ as follows.

%%%%% for compression
\vspace*{-1.0mm}

\begin{definition}[weakly Chaitin $D$-random]
  Let $D\in\R$ with $D\ge 0$,
  and let $\alpha \in\XI$.
  We say that $\alpha$ is \textit{weakly Chaitin $D$-random} if
  there exists $c\in\N$ such that,
  for all $n\in\N^+$,
  $Dn-c \le \K(\alpha_n)$.
  \hfill\QED
\end{definition}

%%%%% for compression
\vspace*{-3.7mm}

\begin{definition}[$D$-compressible]
Let $D\in\R$ with $D\ge 0$,
and let $\alpha \in\XI$.
We say that $\alpha$ is \textit{$D$-compressible} if
$\K(\alpha_n)\le Dn+o(n)$,
which is equivalent to
$\Limsup_{n \to \infty}\K(\alpha_n)/n\le D$.
\hfill\QED
\end{definition}

In the case of $D=1$,
the weak Chaitin $D$-randomness results in the weak Chaitin randomness.
For any $D\in[0,1]$ and any $\alpha\in\XI$,
if $\alpha$ is weakly Chaitin $D$-random and $D$-compressible,
then
\begin{equation}\label{compression-rate}
  \lim_{n\to \infty} \frac{\K(\alpha_n)}{n} = D,
\end{equation}
and therefore
the
compression rate
of $\alpha$ by the program-size complexity $\K$ is
equal to $D$.
Note, however, that \eqref{compression-rate}
does not necessarily implies that $\alpha$ is weakly Chaitin $D$-random.

In the work \cite{T02},
we generalized $\theta$ to $\theta^D$ by
\begin{equation}\label{thetaD}
  \theta^D =
  \sum_{s\in\X}2^{-\frac{\K(s)}{D}}
  \qquad ( D > 0 ).
\end{equation}
Thus,
$\theta=\theta^1$.
If $0<D\le 1$, then $\theta^D$ converges and $0<\theta^D<1$,
since $\theta^D\le \theta<1$.
Theorem \ref{potd} below
was mentioned
in Remark 3.2 of Tadaki \cite{T02}.

%%%%% for compression
\vspace*{-1.0mm}

\begin{theorem}[Tadaki \cite{T02}]\label{potd}
  Let $D\in\R$.
  \begin{enumerate}
    \item If $0<D\le 1$ and $D$ is computable,
      then $\theta^D$ is weakly Chaitin $D$-random.
    \item If $0<D\le 1$ and $D$ is computable,
      then $\theta^D$ is $D$-compressible.
    \item If $1<D$, then $\theta^D$ diverges to $\infty$.
\hfill\QED
  \end{enumerate}
\end{theorem}

\smallskip

%%%%%%%%%%%%%%%%%%%%%%%%%%%%%%%%%%%%%%%%%%%%%%%%%%%%%%%%%%%%%%%%%%%%%%%%%%%
\section{The Shannon entropy of a universal probability}
\label{Shannon}

We say that $p=(p_1,\dots,p_n)$ is a \textit{probability distribution}
if $p_i\in[0,1]$ for all $i=1,\dots,n$ and $p_1+ \dots + p_n=1$.
For any probability distribution $p=(p_1,\dots,p_n)$,
the \textit{Shannon entropy} $H(p)$ of $p$ is defined by
\begin{equation}\label{def-Shannon}
  H(p)=-\sum_{i=1}^n p_i\ln p_i,
\end{equation}
where the $\ln$ denotes the natural logarithm \cite{S48}.
We say that $p=(p_1,\dots,p_n)$ is a \textit{semi-probability distribution}
if $p_i\in[0,1]$ for all $i=1,\dots,n$ and $p_1+ \dots + p_n\le 1$.
We define the Shannon entropy $H(p)$
also for
any
semi-probability distribution $p=(p_1,\dots,p_n)$
by \eqref{def-Shannon}.
Moreover,
for
any
semi-measure $r$ on $\X$,
we define the \textit{Shannon entropy} $H(r)$ of $r$ by
\begin{equation*}
  H(r)=-\sum_{s\in\X} r(s)\ln r(s)
\end{equation*}
in a similar manner to \eqref{def-Shannon}.

In this section,
we prove that the Shannon entropy $H(m)$ of
an arbitrary universal probability $m$ diverges to $\infty$.
For convenience, however,
we first prove the following more general theorem,
Theorem \ref{numerator},
from which the result follows.
For example,
Theorem \ref{numerator} itself can be used to determine
the properties of the notions of \textit{thermodynamic quantities}
introduced by Tadaki \cite{T08CiE}
into algorithmic information theory.

\begin{theorem}\label{numerator}
Let $A$ be an infinite r.e.~subset of $\X$
and let $f\colon\N^+\to\N$ be a total recursive function such that
$\lim_{n\to\infty}f(n)=\infty$.
Then the following hold.
\begin{enumerate}
  \item $\sum_{U(p)\in A}f(\abs{p})2^{-\abs{p}}$ diverges to $\infty$.
  \item If there exists $l_0\in\N^+$ such that
    $f(l)2^{-l}$ is a nonincreasing function of $l$ for all $l\ge l_0$,
    then $\sum_{s\in A} f(\K(s))2^{-\K(s)}$ diverges to $\infty$.
\end{enumerate}
\end{theorem}

\begin{proof}
(i) Contrarily,
assume that $\sum_{U(p)\in A}f(\abs{p})2^{-\abs{p}}$ converges.
Then, there exists $d\in\N^+$ such that
$\sum_{U(p)\in A}f(\abs{p})2^{-\abs{p}}\le d$.
We define the function $r\colon \X\to [0,\infty)$ by
\begin{equation*}
  r(s)=\frac{1}{d}\sum_{U(p)=s}f(\abs{p})2^{-\abs{p}}
\end{equation*}
if $s\in A$; $r(s)=0$ otherwise.
We then see that
$\sum_{s\in\X}r(s)\le 1$
and
therefore $r$ is a lower-computable semi-measure.
Since $P(s)$ is a universal probability by Theorem \ref{eup},
there exists $c\in\N^+$ such that $r(s)\le cP(s)$ for all $s\in\X$.
Hence we have
\begin{equation}\label{fp2p}
  \sum_{U(p)=s}(cd-f(\abs{p}))2^{-\abs{p}}\ge 0
\end{equation}
for all $s\in A$.
On the other hand,
since $A$ is an infinite set and $\lim_{n\to\infty}f(n)=\infty$,
there is $s_0\in A$ such that $f(\abs{p})>cd$ for all $p$ with $U(p)=s_0$.
Therefore we have $\sum_{U(p)=s_0}(cd-f(\abs{p}))2^{-\abs{p}}< 0$.
However, this contradicts \eqref{fp2p},
and the proof of (i) is completed.

(ii) We first note that there is $n_0\in\N$ such that
$\K(s)\ge l_0$ for all $s$ with $\abs{s}\ge n_0$.
Now,
let us assume contrarily that $\sum_{s\in A} f(\K(s))2^{-\K(s)}$ converges.
Then, there exists $d\in\N^+$ such that
$\sum_{s\in A} f(\K(s))2^{-\K(s)}\le d$.
We define the function $r\colon \X\to [0,\infty)$ by
\begin{equation*}
  r(s)=\frac{1}{d}f(\K(s))2^{-\K(s)}
\end{equation*}
if $s\in A$ and $\abs{s}\ge n_0$; $r(s)=0$ otherwise.
We then see that $\sum_{s\in\X}r(s)\le 1$ and
therefore $r$ is a lower-computable semi-measure.
Since $2^{-\K(s)}$ is a universal probability by Theorem \ref{eup},
there exists $c\in\N^+$ such that $r(s)\le c2^{-\K(s)}$ for all $s\in\X$.
Hence, if $s\in A$ and $\abs{s}\ge n_0$,
then $cd\ge f(\K(s))$.
On the other hand,
since $A$ is an infinite set and $\lim_{n\to\infty}f(n)=\infty$,
there is $s_0\in A$ such that $\abs{s_0}\ge n_0$ and $f(\K(s_0))>cd$.
Thus, we have a contradiction, and the proof of (ii) is completed.
\end{proof}

\medskip

From Theorem \ref{numerator} (ii),
we obtain the following result, as desired.

\begin{corollary}\label{ShannonP}
Let $m$ be a universal probability.
Then the Shannon entropy $H(m)$ of $m$ diverges to $\infty$.
\end{corollary}

\begin{proof}
We first note that
there is a real number $x_0>0$ such that
the function $x2^{-x}$ of a real number $x$ is decreasing for $x\ge x_0$.
For this $x_0$,
there is $n_0\in\N$ such that
$-\log_2 m(s)\ge x_0$ for all $s$ with $\abs{s}\ge n_0$.
On the other hand,
by \eqref{eq: K_m},
there is $c\in\N$ such that $-\log_2 m(s)\le \K(s)+c$ for all $s\in\X$.
Thus, we see that
\begin{equation}\label{inese}
\begin{split}
  -&\sum_{s\in \X\text{ \& }\abs{s}\ge n_0} m(s)\log_2 m(s) \\
  &\ge
  \sum_{s\in \X\text{ \& }\abs{s}\ge n_0}(\K(s)+c)2^{-\K(s)-c} \\
  &=
  2^{-c}\sum_{s\in \X\text{ \& }\abs{s}\ge n_0}\K(s)2^{-\K(s)} \\
  &\hspace*{3.8mm}
  +c2^{-c}\sum_{s\in \X\text{ \& }\abs{s}\ge n_0}2^{-\K(s)}.
\end{split}
\end{equation}
Using Theorem \ref{numerator} (ii) with $A=\X$ and $f(n)=n$,
we see that
$\sum_{s\in \X}\K(s)2^{-\K(s)}$ diverges to $\infty$.
It follows from \eqref{inese} that
$-\sum_{s\in \X} m(s)\log_2 m(s)$
also diverges to $\infty$.
This completes the proof.
\end{proof}

\medskip

%%%%%%%%%%%%%%%%%%%%%%%%%%%%%%%%%%%%%%%%%%%%%%%%%%%%%%%%%%%%%%%%%%%%%%%%%%%
\section{The power sum of a universal probability}
\label{power}

In this section,
we investigate the convergence or divergence of
the power sum $\sum_{s\in\X} m(s)^q$ of a universal probability $m$,
and evaluate its degree of randomness
if it converges,
by means of
the notions
of the weak Chaitin $D$-randomness
and the $D$-compressibility.
We first consider
the notion of the weak Chaitin $D$-randomness of the power sum of
a universal probability.
We can generalize Theorem~\ref{potd} (i) and (iii)
on the specific universal probability $2^{-\K(s)}$
over an arbitrary universal probability as follows.

\begin{theorem}\label{psup}
  Let $m$ be a universal probability,
  and let $q\in\R$.
  \begin{enumerate}
    \item If $q\ge 1$ and $q$ is a right-computable real number,
      then $\sum_{s\in\X} m(s)^q$ converges to a left-computable real number
      which is weakly Chaitin $1/q$-random.
    \item If $0<q<1$, then $\sum_{s\in\X} m(s)^q$ diverges to $\infty$.
      \hfill\QED
  \end{enumerate}
\end{theorem}

Theorem \ref{psup} (i) shows that,
for any $q\in\R$ with $q\ge 1$,
the right-computability of $q$
results in
the weak Chaitin $1/q$-randomness of
the power sum $\sum_{s\in\X} m(s)^q$ of
a universal probability $m$.
On the other hand,
Theorem \ref{converse} below shows that
the converse in a certain sense holds.
Theorem \ref{converse} can be proved
based on the techniques used in the proof of
\textit{the fixed point theorem on compression rate}
\cite{T08CiE}.

\begin{theorem}\label{converse}
Let $m$ be a universal probability,
and let $q\in\R$ with $q\ge 1$.
If $\sum_{s\in\X} m(s)^q$ is
a right-computable real number,
then $q$ is weakly Chaitin $1/q$-random.
\hfill\QED
\end{theorem}

Next, we consider
the notion of the $D$-compressibility of the power sum of
a universal probability.
Theorem \ref{potd} (ii) shows that,
for the specific universal probability $m(s)=2^{-\K(s)}$,
if $q$ is a computable real number with $q>1$,
then
the power sum $\sum_{s\in\X} m(s)^q$
is $1/q$-compressible.
Thus,
the following question naturally arises:
Is $\sum_{s\in\X} m(s)^q$ a $1/q$-compressible real number
for any universal probability $m$
and any computable real number $q>1$ ?
As shown in Theorem \ref{ceps},
however,
we can answer this question negatively.

\begin{theorem}\label{ceps}
  There exists a universal probability $m$
  such that,
  for every computable real number $q>1$,
  $\sum_{s\in\X} m(s)^q$ is weakly Chaitin random
  and therefore not $1/q$-compressible.
\end{theorem}

\begin{proof}
We choose any one universal probability $r$,
and then choose any one $c\in\N$ with
$2^{-c}\theta\le r(\lambda)$,
where $\theta$ is defined by \eqref{theta}.
We define the function $m\colon \X\to [0,\infty)$ by
$m(s)=2^{-c}\theta$ if $s=\lambda$; $m(s)=r(s)$ otherwise.
Since $\sum_{s\in\X}r(s)\le 1$,
it follows that $\sum_{s\in\X}m(s)\le 1$.
Therefore,
since $\theta$ is left-computable and
$r$ is a lower-computable semi-measure,
we see that $m$ is a lower-computable semi-measure.
Note
that
$dr(s)\le m(s)$ for all $s\in\X$,
where $d=2^{-c}\theta/r(\lambda)>0$.
Thus, since $r$ is a universal probability,
$m$ is also a universal probability.

On the other hand,
since $\theta$ is weakly Chaitin random,
$m(\lambda)$ is also weakly Chaitin random.
Let $q$ be an arbitrary computable real number with $q>1$.
Then, since $q$ is a computable real number with $q\neq 0$,
it follows that $\K((a^q)_n)=\K(a_n)+O(1)$ for any real number $a>0$.
Thus, $\K((m(\lambda)^q)_n)=\K((m(\lambda))_n)+O(1)$ and therefore
$m(\lambda)^q$ is weakly Chaitin random.
Note that
$\K(a_n)\le \K((a+b)_n)+O(1)$ for any left-computable real numbers $a,b$.
This can be proved using the condition 2 of Lemma 4.4 and
Theorem 4.9 of \cite{CHKW01}.
Thus,
since $m(\lambda)^q$ and $\sum_{s\neq \lambda}m(s)^q$ are left-computable,
we see that $\sum_{s\in\X} m(s)^q$ is weakly Chaitin random.
It follows from $q>1$
that $\sum_{s\in\X} m(s)^q$ is not $1/q$-compressible.
\end{proof}

\medskip

%%%%%%%%%%%%%%%%%%%%%%%%%%%%%%%%%%%%%%%%%%%%%%%%%%%%%%%%%%%%%%%%%%%%%%%%%%%
\section{The Tsallis entropy of a universal probability}
\label{Tsallis}

The notion of Tsallis entropy has been introduced by Tsallis \cite{Ts88}.
Let $q$ be a positive real number with $q\neq 1$.
For any probability distribution $p=(p_1,\dots,p_n)$,
the \textit{Tsallis entropy} $S_q(p)$ of $p$ is defined by
\begin{equation}\label{def-Tsallis}
  S_q(p)=\frac{1-\sum_{i=1}^n p_i^q}{q-1}.
\end{equation}
When $q\to 1$, the Tsallis entropy recovers the Shannon entropy
for any probability distribution.
See \cite{Ts88,GT04} for the detail of
the theory and applications of
Tsallis entropy.

We generalize the definition \eqref{def-Tsallis}
for
any semi-probability distribution $p=(p_1,\dots,p_n)$ by
\begin{equation}\label{def-Tsallis-semi}
  S_q(p)=\frac{\sum_{i=1}^n \{p_i-p_i^q\}}{q-1}.
\end{equation}
In fact,
we see that,
for any semi-probability distribution $p$,
$\lim_{q\to 1}S_q(p)=H(p)$,
and therefore this generalization \eqref{def-Tsallis-semi}
is consistent with the Shannon entropy
for a semi-probability distribution, defined in Section \ref{Shannon}.
Thus,
we define the \textit{Tsallis entropy} $S_q(r)$ of
any semi-measure $r$ on $\X$
by
\begin{equation*}
  S_q(r)=\frac{1}{q-1}\sum_{s\in\X}\{r(s)-r(s)^q\}
\end{equation*}
in a similar manner to \eqref{def-Tsallis-semi}.

In what follows,
we investigate the convergence or divergence of
the Tsallis entropy $S_q(m)$ of a universal probability $m$,
and evaluate its degree of randomness if
it
converges,
in the same manner as the previous section.
We first investigate the convergence and divergence of $S_q(m)$ as follows.

\begin{theorem}\label{tecd}
  Let $m$ be a universal probability,
  and let $q\in\R$.
  \begin{enumerate}
    \item If $q>1$, then $S_q(m)$ converges.
    \item If $0<q<1$, then $S_q(m)$ diverges to $\infty$.
  \end{enumerate}
\end{theorem}

\begin{proof}
  Theorem \ref{tecd} follows immediately from Theorem \ref{psup}.
\end{proof}

\medskip

Theorem \ref{tewcr} below shows that,
if the total sum of a universal probability $m$ is small,
then the Tsallis entropy of $m$ has to be maximally random
with respect to the degree of randomness.

\begin{theorem}\label{tewcr}
  Let $m$ be a universal probability,
  and let $q$ be a computable real number with $q>1$.
  If $m(s)\le q^{\frac{1}{1-q}}$ for all $s\in\X$,
  then $S_q(m)$ is left-computable and weakly Chaitin random.
\end{theorem}

\begin{proof}
By Theorem \ref{tecd} (i),
there is $d\in\N^+$ such that $S_q(m)\le d$.
We define $r\colon \X\to(0,\infty)$ by
$r(s)=F(m(s))/d$,
where $F\colon(0,1]\to[0,\infty)$ with $F(x)=(x-x^q)/(q-1)$.
We show that $r$ is a universal probability.

Obviously, $\sum_{s\in \X}r(s)\le 1$.
Since $m$ is a lower-computable semi-measure,
there exists a total recursive function
$f\colon\N^+\times \X\to\Q$
such that, for each $s\in \X$,
$\lim_{n\to\infty} f(n,s)=m(s)$ and
$\forall\,n\in\N^+\;\>0<f(n,s)\le m(s)$.
Since $F(x)$ is continuous and increasing
for all $x\in(0,q^{\frac{1}{1-q}}]$,
it follows that,
for each $s\in \X$,
$\lim_{n\to\infty} F(f(n,s))=F(m(s))$ and
$\forall\,n\in\N^+\;\>0\le F(f(n,s))\le F(m(s))$.
On the other hand,
since
$q$ is computable,
there exists a total recursive function
$g\colon\N^+\times \X\to\Q\cap[0,\infty)$
such that, for each $s\in \X$ and each $n\in\N^+$,
\begin{equation*}
  F(f(n,s))-2^{-n}\le g(n,s)\le F(f(n,s)).
\end{equation*}
Hence, $r$ is a lower-computable semi-measure.
Note that
$x/q\le F(x)$ for all $x\in(0,q^{\frac{1}{1-q}}]$.
It follows that
$m(s)/(qd)\le r(s)$ for all $s\in\X$.
Thus, since $m$ is a universal probability,
$r$ is also a universal probability.

It follows from Theorem \ref{psup} (i) that
$\sum_{s\in\X}r(s)=S_q(m)/d$ is weakly Chaitin random.
Note that
$\K(a_n)\le \K((ab)_n)+O(1)$
for any left-computable real numbers $a,b>0$.
This can be proved using the condition 4 of Lemma 4.4 and
Theorem 4.9 of \cite{CHKW01}.
Thus,
since $\sum_{s\in\X}r(s)$ and $d$ are
left-computable positive real numbers,
we see that $S_q(m)$ is weakly Chaitin random
and, obviously, left-computable.
\end{proof}

\medskip

Based on Theorem \ref{tewcr},
we can show a stronger result than Theorem \ref{tewcr}
with respect to
the range of
the degree of randomness of
the Tsallis entropy $S_q(m)$.
Theorem \ref{sqy} and Corollary~\ref{sqwcdrdc} below show that
the Tsallis entropy of a universal probability can have
any computable degree of randomness $D$.
Note, however, that
Theorem \ref{sqy} is not a generalization of Theorem \ref{tewcr}.
The reason is as follows:
The Tsallis entropy $S_q(m)$ is right-computable in Theorem \ref{sqy}
whereas it is not right-computable in Theorem \ref{tewcr}.

\begin{theorem}\label{sqy}
Let $q$ be a computable real number with $q>1$.
Then, for any right-computable real number $y\in(0,q^{\frac{q}{1-q}}]$,
there exists a universal probability $m$ such that $S_q(m)=y$.
\end{theorem}

\begin{proof}
Let $F\colon(0,1]\to[0,\infty)$ with $F(x)=(x-x^q)/(q-1)$,
and let $x_0$ be the unique real number such that
$q^{\frac{1}{1-q}}<x_0<1$ and $F(x_0)=y/2$.
We choose any one rational number $c$ such that
$0<c\le\min\{q^{\frac{1}{1-q}},1-x_0,(q-1)y/2\}$.
We also choose any one universal probability $r$.
We then define a universal probability
%%%%% For impression
$r_1\colon \X$ $\to(0,1)$
by $r_1(s)=cr(s)$.
Since $r_1(s)\le q^{\frac{1}{1-q}}$ for all $s\in\X$,
it follows from Theorem \ref{tewcr} that
$S_q(r_1)$ is left-computable.

Let $\Theta=S_q(r_1)$.
From $\sum_{s\in\X}r(s)\le 1$
we have $\sum_{s\in\X}r_1(s)\le c$.
Therefore,
\begin{equation*}
  \Theta =
  \sum_{s\in\X}F(r_1(s)) <
  \frac{1}{q-1}\sum_{s\in\X}r_1(s) \le
  \frac{c}{q-1}.
\end{equation*}
Since
$c/(q-1)\le y/2$,
it follows that $y/2<y-\Theta<y$.

Note that
$F(x)$ is continuous and decreasing
for all $x\in[q^{\frac{1}{1-q}},1]$.
Thus,
since $F(q^{\frac{1}{1-q}})=q^{\frac{q}{1-q}}\ge y$
and $y/2>F(1)=0$,
there exists the unique real number $a$ such that
$q^{\frac{1}{1-q}}<a<x_0$
and $F(a)=y-\Theta$.
We see that $a$ is left-computable.
This is because $y-\Theta$ is right-computable,
$q$ is computable,
and $F(x)$ is decreasing
for all $x\in(q^{\frac{1}{1-q}},x_0)$.

We define the function $m\colon \X\to (0,\infty)$ by
$m(s)=a$ if $s=\lambda$; $m(s)=r_1(s-1)$ otherwise.
Note here that $\X$ is identified with $\N$.
Then, it follows from $c\le 1-x_0$ and $a<x_0$ that
$\sum_{s\in\X}m(s)<1$.
Thus,
since $r_1$ is a lower-computable semi-measure and $a$ is left-computable,
we see that $m$ is a lower-computable semi-measure.
Since $r_1$ is a universal probability and $a>0$,
we further see that
$m$ is a universal probability.
On the other hand,
$S_q(m)=F(a)+S_q(r_1)=F(a)+\Theta=y$.
This completes the proof.
\end{proof}

\begin{corollary}\label{sqwcdrdc}
Let $q$ be a computable real number with $q>1$.
Then, for any computable real number $D\in[0,1]$,
there exists a universal probability $m$ such that
$S_q(m)$ is weakly Chaitin $D$-random and $D$-compressible.
\end{corollary}

\begin{proof}
In the case of $D=0$,
consider a rational number $y\in(0,q^{\frac{q}{1-q}}]$ in Theorem~\ref{sqy}.
In the case of $D>0$,
consider $y=a(1-\theta^D)$ in Theorem~\ref{sqy},
where $a$ is any one rational number with $a\in(0,q^{\frac{q}{1-q}}]$
and $\theta^D$ is defined by \eqref{thetaD}.
In this case,
the result follows from Theorem \ref{potd} (i) and (ii).
\end{proof}

\medskip

%%%%%%%%%%%%%%%%%%%%%%%%%%%%%%%%%%%%%%%%%%%%%%%%%%%%%%%%%%%%%%%%%%%%%%%%%%
\section{Conclusion}
\label{conclusion}

In this paper,
we have investigated the properties of
the Shannon entropy, the power sum, and the Tsallis entropy of
a universal probability,
from the point of view of algorithmic randomness.
Future work may
aim at generalizing R\'eny entropy over a universal probability
properly
and investigating its randomness properties.

\medskip

%%%%%%%%%%%%%%%%%%%%%%%%%%%%%%%%%%%%%%%%%%%%%%%%%%%%%%%%%%%%%%%%%%%%%%%%%%
\section*{Acknowledgments}

This work was supported both
by
SCOPE
(Strategic Information and Communications R\&D Promotion Programme)
from the Ministry of Internal Affairs and Communications
of Japan
and
by
KAKENHI,
Grant-in-Aid for Scientific Research (C)
(20540134).



% http://www.ctan.org/tex-archive/macros/latex/contrib/supported/IEEEtran/bibtex
%\bibliographystyle{IEEEtran.bst}
%\begin{thebibliography}{1}

%%%%%%%%%%%%%%%%%%%%%%%%%%%%%%%%%%%%%%%%%%%%%%%%%%%%%%%%%%%%%%%%%%%%%%%%%%
\end{document}